\pdfoutput=1
\RequirePackage{ifpdf}
\ifpdf % We are running pdfTeX in pdf mode
\documentclass[pdftex]{sigma}
\else
\documentclass{sigma}
\fi

\newtheorem{teo}{Theorem}[section]

\newtheorem{prop}[teo]{Proposition}

{\theoremstyle{definition}

\newtheorem{rmk}[teo]{Remark }
\newtheorem{exm}[teo]{Example}}

\numberwithin{equation}{section}

\begin{document}
%\allowdisplaybreaks

\newcommand{\arXivNumber}{1808.01889}

\renewcommand{\PaperNumber}{013}

\FirstPageHeading

\ShortArticleName{Block-Separation of Variables: a Form of Partial Separation for Natural Hamiltonians}

\ArticleName{Block-Separation of Variables: a Form of Partial\\ Separation for Natural Hamiltonians}

\Author{Claudia Maria CHANU and Giovanni RASTELLI}

\AuthorNameForHeading{C.M.~Chanu and G.~Rastelli}

\Address{Dipartimento di Matematica, Universit\`a di Torino, Torino, via Carlo Alberto 10, Italy}
\Email{\href{mailto:claudiamaria.chanu@unito.it}{claudiamaria.chanu@unito.it}, \href{giovanni.rastelli@unito.it}{giovanni.rastelli@unito.it}}

\ArticleDates{Received August 07, 2018, in final form February 14, 2019; Published online February 23, 2019}

\Abstract{We study twisted products $H=\alpha^rH_r$ of natural autonomous Hamiltonians $H_r$, each one depending on a separate set, called here separate $r$-block, of variables. We show that, when the twist functions $\alpha^r$ are a row of the inverse of a block-St\"ackel matrix, the dynamics of $H$ reduces to the dynamics of the $H_r$, modified by a scalar potential depending only on variables of the corresponding $r$-block.
It is a kind of partial separation of variables. We characterize this block-separation in an invariant way by writing in block-form classical results of St\"ackel separation of variables. We classify the
block-separable coordinates of $\mathbb E^3$.}


\Classification{70H05; 37J15; 70H06}

\section{Introduction}

In \cite{St} Paul St\"ackel started the study of (complete) separation of variables in orthogonal coordinates for the Hamilton--Jacobi equation of natural Hamiltonians with $N$ degrees of freedom. The characterization given by St\"ackel is both coordinate dependent~-- involving $N\times N$~St\"ackel matrices (see Theorem~\ref{teo1} below)~-- and invariant~-- involving $N$ quadratic first integrals of the Hamiltonian. In the following years, the theory was widely developed by Levi-Civita \cite{LC}, Eisenhart \cite{Ei,Ei1} and many others (see \cite{Ka,18bis,MS} for more complete references). We point out that complete separation implies the completeness of the separated integral of the Hamilton--Jacobi equation, because of the existence of $N$ independent constants of motion in involution and, consequently, the Liouville integrability of the Hamiltonian system.

St\"ackel himself considered in \cite{St97} the case of partial separation of variables and obtained a~sufficient characterization of it in terms of St\"ackel matrices of reduced dimension and of a~corresponding number of quadratic first integrals of the Hamiltonian. Partial separation of variables gained much less interest than complete separation. This fact is certainly related with the little use of partial separation in the search of solutions of the Hamilton--Jacobi equation. Moreover, partial separation does not guarantee the existence of complete integrals of the partially sepa\-rated Hamilton--Jacobi equation, neither Liouville integrability. Nevertheless, as pointed out by~\cite{Le}, in this case the Jacobi method of inversion can, sometimes, produce additional first integrals of the Hamiltonian. Recent papers develop partial separation theory for Hamilton--Jacobi and Schr\"odinger equations, improving somehow the results of St\"ackel, by giving a more detailed characterization of the metric coefficients in partially separable coordinates and by providing further conditions for the separation of the quantum systems \cite{Ha, Ma}. Partial separation of Hamilton--Jacobi and Helmholtz equations on four-dimensional manifolds is briefly considered in~\cite{BKM}. A~different approach to non-complete additive separation is represented by non-regular separation which relies on the existence of an additively separated solution on proper submanifolds only~\cite{BCM, C,18bis}.

\looseness=1 In our study, we shift for the first time the interest from the Hamilton--Jacobi equation to the dynamics of the system. We observe that the partial separation introduced by St\"ackel, as well as the complete separation, establishes a dynamical relationship between~$H$ and the (partially) separated equations when these are considered as Hamiltonians on submanifolds of the original phase space. Namely, we find that the projection of the orbits of~$H$ on these submanifolds, spanned by the separated blocks of coordinates, coincides with the orbits of the separated Hamiltonians on the same submanifolds. The only difference is a position-dependent rescaling of the corresponding Hamiltonian parameters. As a consequence, the dynamics of~$H$ can be decomposed into a number of lower-dimensional Hamiltonian systems, allowing in some case a~simpler analysis of the original system. The separated blocks of coordinates, considered together, form a $N$-dimensional coordinate system on the base manifold of~$H$ and in these coordinates the $N$-dimensional metric tensor takes a block-diagonal form. This fact motivates the name we choose for this kind of separation. We prefer to not use the expression {\it partial separation} since it is already associated with Hamilton--Jacobi theory, which we do not consider here.
By shifting the focus from Hamilton--Jacobi theory to the dynamics, we remove the obstruction represented by the completeness of the integral of the Hamilton--Jacobi equation, which is strictly connected with St\"ackel theory of complete separation of variables. Indeed, partial separation does not imply the existence of a complete integral, so that the Jacobi method (the construction of a~canonical transformation to a~trivially integrable Hamiltonian) generally fails. On the contrary, our dynamical interpretation of block separation is basically insensitive to complete or partial separation. We find useful and natural to relate block-separation with the structure of {\it twisted product} that the Hamiltonian assumes when the separation is possible. This allows us to state our results in a form very close to analogous results in classical St\"ackel separation, analogy missing in all the other studies about partial separation. Namely, we can characterize block-separation by introducing ``block'' versions of celebrated Levi-Civita and Eisenhart equations, and of more recent theorems about complete separation.

The main result of block-separation is the splitting of a $N$-dimensional Hamiltonian system into lower-dimensional systems (the blocks). Thus, methods of analysis disposable only for low-dimensional systems become available, such as, for example, the topological classification of integrable Hamiltonian systems \cite{BF}.

In Section~\ref{section2} we recall the basic theorems about St\"ackel complete separation of variables which we rewrite in block-separable form. In Section~\ref{section3} we define twisted products of Hamiltonians and state some relevant properties of them. In Section~\ref{section4} we give a dynamical interpretation of St\"ackel separation, providing examples of the related properties of time-scaling. In Section~\ref{section5} we introduce block-separation and our main results about its characterization, with the block-like formulations of Levi-Civita, Eisenhart and other theorems, and we provide an invariant characterization of block separation. The explicit example we deal with is the four-body Calogero system. In Section~\ref{section6}, we characterize, at least with necessary conditions, all the possible block-separable coordinates of $\mathbb E^3$. Section~\ref{section7}
contains our final considerations and comments.

\section{Outline of St\"ackel separation}\label{section2}

We briefly recall the principal theorems regarding complete separation of the Hamilton--Jacobi equation, see \cite{Ka} and \cite{Be} for further details. The theory of complete additive separation of the Hamilton--Jacobi equation begins with the work of St\"ackel \cite{St,St97} about separability of the Hamilton--Jacobi equation in orthogonal coordinates. The Einstein summation convention on equal indices is understood, unless otherwise stated.

\begin{teo} \label{teo1} In a given orthogonal coordinate system $(q^i)$, the Hamilton--Jacobi equation
\begin{gather}\label{HJ}
H=\frac 12\left(g^{ii}\left(\frac {\partial W}{\partial q^i}\right)^2+V(q)\right)=c_1,
\end{gather}
admits a complete integral in the separated form,
 \begin{gather*}
W=\sum_{i=1}^NW_i\big(q^i,c_1,c_a\big), \qquad a=2,\dots,N, \qquad \det \left(\frac{\partial^2W}{\partial q^i \partial c_j}\right)\neq 0,\qquad c_j=(c_1,c_a),
 \end{gather*}
if and only if
\begin{enumerate}\itemsep=0pt
\item[$1)$] there exists a $N \times N$ matrix $S$ which is invertible and such that each element of its $j$-th row depends on $q^j$ only, such that the $\big(g^{ii}\big)$ are a~row of $S^{-1}$. The matrix $S$ is called St\"ackel matrix.
\item[$2)$] $V$ is a St\"ackel multiplier, i.e., there exist $N$ functions $v_i\big(q^i\big)$ such that
 \begin{gather*}
V=v_i\big(q^i\big)g^{ii}.
 \end{gather*}
\end{enumerate}
As a consequence of $(1)$ and $(2)$, there exist $N-1$ independent quadratic first integrals $(K_a)$ of $H$ $(a=2,\dots,N)$ such that $(c_i)=(c_1,c_a)$ are the constant values of $(H,K_a)$, where $(K_a)=K_2, \dots, K_{N}$.

A complete integral of \eqref{HJ} is then determined by the $N$ separated equations
\begin{gather*}\label{seq}
\left(\frac{{\rm d}W_r}{{\rm d}q^r}\right)^2+v_r=2S^i_rc_i,
\end{gather*}
where all $S^i_r$ depend $(i=1,\dots,N)$ on the coordinate $q^r$ only.
\end{teo}

We remark that completeness for the integral $W$ of the Hamilton--Jacobi equation means that it depends on $N$ parameters $(c_j)$, constants of motion, such that
\begin{gather*}
\det\left(\frac{\partial W}{\partial q^i\partial c_j}
\right)\neq 0.
 \end{gather*}
Therefore, the $(c_j)$ can be part of a new set of canonical coordinates in which the Hamiltonian flow becomes trivially integrable.

Later, Levi-Civita \cite{LC} obtained necessary and sufficient conditions for the complete separability of a generic Hamiltonian in a general coordinate system. For natural Hamiltonians and orthogonal coordinates, $H=\frac 12 g^{ii}p_i^2+V(q)$, the Levi-Civita equations
split into
\begin{gather}\label{LC}
g^{ii}g^{jj}\partial_{ij}g^{kk}-g^{ii}\partial_ig^{jj} \partial_jg^{kk}-g^{jj}\partial_jg^{ii}\partial_ig^{kk}=0,
\end{gather}
with $i\neq j$ not summed, $i,j,k=1,\dots,N$, for the components of the metric tensor, and
 \begin{gather*}
g^{ii}g^{jj}\partial_{ij}V-g^{ii}\partial_ig^{jj}\partial_jV-g^{jj}\partial_jg^{ii}\partial_iV=0,
 \end{gather*}
with $i\neq j$ not summed, $i,j=1,\dots,N$, called {\em Bertrand--Darboux equations} and whose solution (for $g^{ii}$ satisfying~(\ref{LC})) is $V$ in the form of a St\"ackel multiplier.

In~\cite{Ei}, Eisenhart provided a geometrical characterization of complete separation of variables in orthogonal coordinates introducing Killing tensors and, later, determined all the possible orthogonal separable coordinate systems of~$\mathbb E^3$. He determined eleven types of orthogonal separable coordinate systems, which are described, for example, in~\cite{MS}. The fundamental Eisenhart equations
 \begin{gather*}
\partial_i\lambda_j=(\lambda_i-\lambda_j)\partial_i\ln \big|g^{jj}\big|,\qquad i,j =1, \dots, N
 \end{gather*}
characterize the eigenvalues of a Killing tensor: let $K$ be a 2-tensor with eigenvalues $(\lambda_i)$ and eigenvectors $(\partial_i)$, then $K$ is a Killing tensor if and only if the Eisenhart equations are satisfied. The integrability conditions of the Eisenhart equations for a Killing tensor with simple eigenvalues coincide with the Levi-Civita equations~(\ref{LC}).

The coordinate systems can be geometrically understood as foliations of hypersurfaces called {\it coordinate webs}. The separability of a coordinate web can be characterized by a single {\it characteristic Killing tensor}, i.e., a symmetric Killing 2-tensor with pointwise simple eigenvalues and normally integrable eigenvectors, which determine in each point of the space (up to possible singular sets of zero measure) the basis of coordinate vectors \cite{Be}. We recall that a symmetric {\it Killing $2$-tensor} $K$ is defined by the equivalent equations
 \begin{gather*}[g,K]=0, \qquad \nabla_{(i}K_{jk)}=0, \end{gather*}
where $[\cdot,\cdot]$ is the Schouten bracket and $\nabla$ is the covariant derivative with respect to the metric~$g$.

\begin{teo}\label{t12} The Hamilton--Jacobi equation of a natural Hamiltonian with scalar potential~$V$ is completely separable in an orthogonal coordinate web if and only if there exists a characteristic Killing $2$-tensor $K$ whose eigenvectors are normal to the foliations of the web and such that ${\rm d}(K{\rm d}V)=0$.
\end{teo}

The last condition is equivalent to say that $V$ is a~St\"ackel multiplier in the orthogonal coordinates associated with $K$. In the formula of Theorem~\ref{t12} $K$ is considered as a linear operator mapping one-forms into one-forms.

Necessary and sufficient conditions for a Killing tensor to be characteristic are given in~\cite{ AMS}.

\begin{teo}[Tonolo--Schouten--Nijenhuis \cite{Nij,Sc,To}] A $2$-tensor $K$ with real distinct eigenvalues has normal eigenvectors if and only if the following conditions
are satisfied
\begin{gather*}
N^l_{[i j} g_{k]l} = 0,\qquad N^l_{[i j} K_{k]l} = 0, \qquad N^l_{[i j} K_{k]m}K^m_l= 0,
\end{gather*}
where $N^i_{jk}$ are the components of the Nijenhuis tensor of $K^i_j$ defined by
 \begin{gather*}
N^i_{jk} = K^i_lK_{[ j,k]} + K^l_{[ j} K^i_{k],l}.
 \end{gather*}
\end{teo}

An equivalent formulation of Theorem \ref{t12} involves $N$ independent quadratic first integrals, therefore, $N$ Killing 2-tensors, instead of a single characteristic Killing tensor \cite{Be}.

\begin{teo} The natural Hamiltonian $H$ is separable in some orthogonal coordinates $\big(q^{i}\big)$, if and only if
\begin{enumerate}\itemsep=0pt
\item[$1)$] there exist other $N-1$ independent quadratic in the momenta functions $K_a$ such that
 \begin{gather*}
\{H,K_a\}=0,
 \end{gather*}
\item[$2)$] the Killing two-tensors $(k_a)$ associated with $(K_a)$ are simultaneously diagonalized with pointwise independent eigenvalues and have common normally integrable eigevectors.
\end{enumerate}
It follows that $\{K_a,K_b\}=0$.
\end{teo}

The original formulation of the theorem requires the reality of the eigenvalues, however, this request is unnecessary if one accepts also complex separable coordinates~\cite{DR}.

\section{Twisted products of Hamiltonians}\label{section3}

Let
 \begin{gather*}
M=\times _{r=1}^nM_r,
 \end{gather*}
be the product of $n$ Riemannian or pseudo-Riemannian manifolds $(M_r,g_r)$ of dimension $n_r$, so that dim$(M)=n_1+\dots+n_n=N$, and let $\alpha^r$ be $n$ non zero functions on $M$. The manifold $M$ with metric tensor
 \begin{gather*}
G=\alpha^1 g_1+\cdots+ \alpha^ng_n,
 \end{gather*}
is a Riemannian or pseudo-Riemannian manifold called {\it twisted product manifold} of the $(M_r)$ with twist functions $(\alpha^r)$ \cite{MRS99}. In the case when $\alpha^1=1$ and $\alpha^2,\dots,\alpha^n$ are functions on $M_1$, the manifold $M$ is called {\it warped product}. We extend to functions on $T^*M$ and $T^*M_r$ the concept of twisted and warped products in a natural way. In particular, for each $r$ we consider natural Hamiltonians
 \begin{gather*}
H_r=\frac 12 \big(g_r^{r_ir_j}p_{r_i}p_{r_j}+V_r\big(q^{r_i}\big)\big),
 \end{gather*}
where $\big(q^{r_i},p_{r_i}\big)$, $i=1,\dots,n_r$, are canonical coordinates on $T^*M_r$, we construct {\it twisted product}
 \begin{gather*}
H=\alpha^rH_r=\alpha^1 H_1+ \dots +\alpha^nH_n,
 \end{gather*}
of the $H_r$ with twist functions $\alpha^r\in \mathcal F(M)$.

Then, $H$ is a natural Hamiltonian on $T^*M$ with metric $G$ and potential $V=\alpha^rV_r$.

The manifold $M$ is naturally endowed with block-diagonal coordinates $(q^{r_i})$ such that the components of $G$ are in the form
 \begin{gather*}
G^{r_ir_j}=\alpha^rg^{r_ir_j}_r,\qquad G^{r_is_j}=0, \qquad s\neq r,
 \end{gather*}
we call these coordinates {\it twisted coordinates}.

We have now $n+1$ Hamiltonians, each one with its own Hamiltonian parameter. We call~$t$ the Hamiltonian parameter of $H$ and $\tau_r$ the Hamiltonian parameter of $H_r$. From Hamilton's equations we get
 \begin{gather*}
\frac {{\rm d}q^{r_i}}{{\rm d}t}=\frac{\partial H}{\partial p_{r_i}}=\alpha^r\frac{\partial H_r}{\partial p_{r_i}}=\alpha^r \frac {{\rm d}q^{r_i}}{{\rm d}\tau_r},
 \end{gather*}
and
 \begin{gather*}
\frac {{\rm d}p_{r_i}}{{\rm d}t}=-\frac{\partial H}{\partial q^{r_i}}=-\alpha^r\frac{\partial H_r}{\partial q^{r_i}}-H_s\frac{\partial \alpha^s}{\partial q^{r_i}}=\alpha^r \frac {{\rm d}p_{r_i}}{{\rm d}\tau_r}-H_s\frac{\partial \alpha^s}{\partial q^{r_i}}.
 \end{gather*}
 Therefore, the relation between the Hamiltonian vector fields $X_H$ of $H$ and $X_{r}$ of $H_r$ is
 \begin{gather*}
X_H=\bar X_1+ \dots+\bar X_N-H_s\frac{\partial \alpha^s}{\partial q^{r_i}}\frac{\partial}{\partial p_{r_i}},
 \end{gather*}
where
 \begin{gather*}
\bar X_r=\alpha^rX_r,\qquad r\; \text{not summed},
 \end{gather*}
is the rescaled Hamiltonian vector field of $H_r$.

\section{St\"ackel systems as twisted Hamiltonians}\label{section4}

In this section we study St\"ackel systems in their nature of twisted Hamiltonians. Our aim is to enlighten the relations among the dynamics of the $N$-dimensional Hamiltonian system determined by $H$ and the dynamics of the $N$ one-dimensional Hamiltonians $H_r$, so that $n_1=\dots =n_N=1$, determined by the separated equations of $H$. Separation of variables for the Hamilton--Jacobi equation of $H$ will not be of primary interest in what follows. See Section~\ref{section2} for definitions of St\"ackel matrix, St\"ackel multiplier and separated equations.

Let be $H_r=\frac 12 (p_r^2+V_r)$ and assume that $\alpha^r$ are a row (say the first one) of the inverse of a St\"ackel matrix $S$ for given coordinates $(q^{r})$. Then, the twisted product $H=\alpha^rH_r$ admits separation of variables and we have the separated equations
\begin{gather}\label{se}
H_r=S^i_rc_i,
\end{gather}
where $c_i$ are $N$ constants, corresponding to the $N$ constants of motion $K_i$ of $H=K_1$ and
 \begin{gather*}
K_a=\big(S^{-1}\big)_a^rH_r,\qquad a=2,\dots,N.
 \end{gather*}
The Hamilton--Jacobi complete separated integral $W=W_1\big(q^1,c_i\big)+\dots +W_N\big(q^N,c_i\big)$ is given by integration of
\begin{gather*}\label{hje}
\left(\frac{{\rm d} W_r}{{\rm d}q^r}\right)^2+V_r= 2S_r^ic_i.
\end{gather*}

The Hamilton's equations of $H$, in time $t$, are
\begin{gather}
\dot q^r=\alpha^rp_r,\nonumber\\
 \dot p_r=-\partial_r \alpha^iH_i-\alpha^i\partial_rH_i=-\partial_r\alpha^iS_i^jc_j-\frac 12 \alpha^r \frac {\rm d}{{\rm d}q^r}V_r,\label{he1}
 \end{gather}
where we use the separated equations (\ref{se}) to replace $H_r$ along the integral curves. Since $\alpha^r=\big(S^{-1}\big)^r_1$ is a row of the inverse of $S$, we have
\begin{gather}\label{5bis}
\partial_r\alpha^rS_i^j=\partial_r\big(\alpha^iS_i^j\big)-\alpha^i\partial_rS^j_i=\partial_r\big(\delta^j_1\big)-\alpha^r\frac{{\rm d}}{{\rm d}q^r}S^j_r=-\alpha^r\frac{{\rm d}S^j_r}{{\rm d}q^r},
\end{gather}
and the same for all other elements of the rows of $S^{-1}$.
Then, we can write
\begin{gather}\label{he2}
\dot p_r=\alpha^r\frac{{\rm d}}{{\rm d}q^r}\left(c_jS^j_r-\frac 12 V_r\right).
\end{gather}
Let $\gamma_P$ be the integral curve of $X_H$ containing a point $P\in T^*M$. We consider the values $c_i=K_i(P)$ and we introduce the Hamiltonians
 \begin{gather*}
\tilde H_r=H_r-c_jS^j_r,
 \end{gather*}
with Hamiltonian parameters $\tilde \tau_r$. We can write the equations of Hamilton for $\tilde H_r$ as
\begin{gather}
\frac{{\rm d}}{{\rm d}\tilde \tau_r} q^r = p_r=\partial^r\tilde H_r,\label{tt1}\\
\frac {\rm d}{{\rm d}\tilde \tau_r} p_r = \frac{{\rm d}}{{\rm d}q^r}\left(c_jS^j_r-\frac 12 V_r\right)=-\partial_r \tilde H_r.\label{tt2}
\end{gather}
Therefore,
\begin{prop}\label{p0} For each orbit of $H$, the $N$ Hamiltonian vector fields $X_{\tilde H_r}$ of the $\tilde H_r$ are proportional to the components with respect to $(\partial_r, \partial^r)$ of the Hamiltonian vector field $X_H$ of $H$, where $\alpha^r$ are the proportionality functions.
\end{prop}

\begin{proof}Let be
 \begin{gather*}X_{\tilde H_r}=\frac {\partial \tilde H_r}{\partial p_r}\partial_r-\frac {\partial \tilde H_r}{\partial q^r}\partial^r,
 \end{gather*}
and
 \begin{gather*}(X_H)_r=\dot q^r\partial_r+\dot p_r\partial^r. \end{gather*}

Due to (\ref{tt1}) and (\ref{tt2}) we can write (\ref{he1}) and (\ref{he2}) as
\begin{gather*}
\dot q^r=\alpha^r\partial^r\tilde H_r,\qquad \dot p_r=-\alpha^r\partial_r \tilde H_r,
\end{gather*}
and it follows immediately
 \begin{gather*}(X_H)_r=\alpha^r X_{\tilde H_r}.\tag*{\qed} \end{gather*}\renewcommand{\qed}{}
\end{proof}

\begin{rmk} After Proposition \ref{p0} we can put
\begin{gather*}\label{1a}
\alpha^r=\frac{{\rm d}\tilde \tau_r}{{\rm d}t},
\end{gather*}
and consider the twist functions as determining position-dependent time-scalings between the Hamiltonian parameters $t$ and $\tilde \tau_r$.
\end{rmk}

From Proposition \ref{p0} follows the important result

\begin{prop}\label{p1}
The projection of each orbit of $H$ on each coordinate manifold $\big(q^r,p_r\big)$ coincides with the orbit of $\tilde H_r=H_r-c_jS^j_r$.
\end{prop}

\begin{rmk} The Lagrange equations of the dynamics of $\tilde H_r$, expressed in times $\tilde \tau_r$ are
\begin{gather*}\label{sem}
\frac{{\rm d}^2q^r}{{\rm d}\tilde \tau_r^2}=\frac{{\rm d}}{{\rm d}q^r}\left(c_jS^j_r-\frac 12 V_r\right).
\end{gather*}
\end{rmk}

\begin{rmk} Observe that
 \begin{gather*}
\tilde H=\alpha^r\tilde H_r=H-c_1, \qquad \tilde K_j=\big(S^{-1}\big)^r_j\tilde H_r=K_j-c_j, \qquad j \neq 1,
 \end{gather*}
i.e., the St\"ackel systems associated with $H$ and $\tilde H$ coincide up to additive constants. To the constants $(c_i)$ for $(H,K_a)$ correspond the costants $(\tilde c_i=0)$ for $\big(\tilde H, \tilde K_a\big)$.
\end{rmk}

\begin{exm}{\it Twisted product of pendula.} In order to show the effect of the time-scaling described above, we consider the twisted product of the following three one-dimensional Hamiltonians
 \begin{gather*}
H_i=\frac 12 \big(p_i^2-\cos q^i\big), \qquad i=1,2, \qquad H_3=\frac 12 p_3^2,
 \end{gather*}
corresponding to two pendula and a purely inertial term, coupled together by the first row of the inverse of the $3\times 3$ matrix
 \begin{gather*}
S=\left(\begin{matrix} 2 && 1+q^1 && 2 \big(q^1\big)^2+2 \\
 3 && q^2 && \big(q^2\big)^3+2 \\
 4 && q^3 && \big(q^3\big)^2+1 \end{matrix}\right),
 \end{gather*}
which is a St\"ackel matrix in a neighborhood of the origin, since the Taylor expansion up to the second order terms of its determinant $\Delta$ around $(0,0,0)$ is $\Delta=5+5q^1-6 q^2+2 q^3$. The elements of the matrix $S^{-1}$ are therefore quite complicated rational functions that we do not need to compute explicitly but make the coupling of the $H_r$ suitable to enhance the effect of the time-scaling. We take as $(\alpha^r)$ the first row of $S^{-1}$ and consider
 \begin{gather*}
H=\alpha^rH_r=c_1.
 \end{gather*}
The quadratic first integrals of $H$ are determined by the remaining rows of $\big(S^{-1}\big)$
 \begin{gather*}
K_a=\big(S^{-1}\big)_a^rH_r=c_a,\qquad a=2,3.
 \end{gather*}

We already know from the previous section that, despite the complicated expression of the coupling terms $(\alpha^r)$, the relation among the dynamics of $H$ and of the separated Hamiltonians $\tilde H_r=H_r-c_aS^a_r$ reduces to a simple position-dependent time scaling.

We plot the numerical evaluation of the systems of Hamiltonian $H$ and
 \begin{gather*}\tilde H_1=H_1-2c_1-\big(1+q^1\big) c_2-2\big(\big(q^1\big)^2+1\big)c_3
 \end{gather*}
respectively, and project the orbits on $\big(q^1,p_1\big)$, obtaining, for the initial conditions $p_1 = 0$, $p_2 = 0$, $p_3 = 0$, $q^1 = 0.2$, $q^2 = -0.2$, $q^3 = 0$ and consequently $c_1=H = 0.09494666248$, $c_2 = 0.0916913483$, $ c_3 = -0.3797866499$, the graphs in Fig.~\ref{Fig1},
\begin{figure}[t]\centering
\includegraphics[width= 4.5cm]{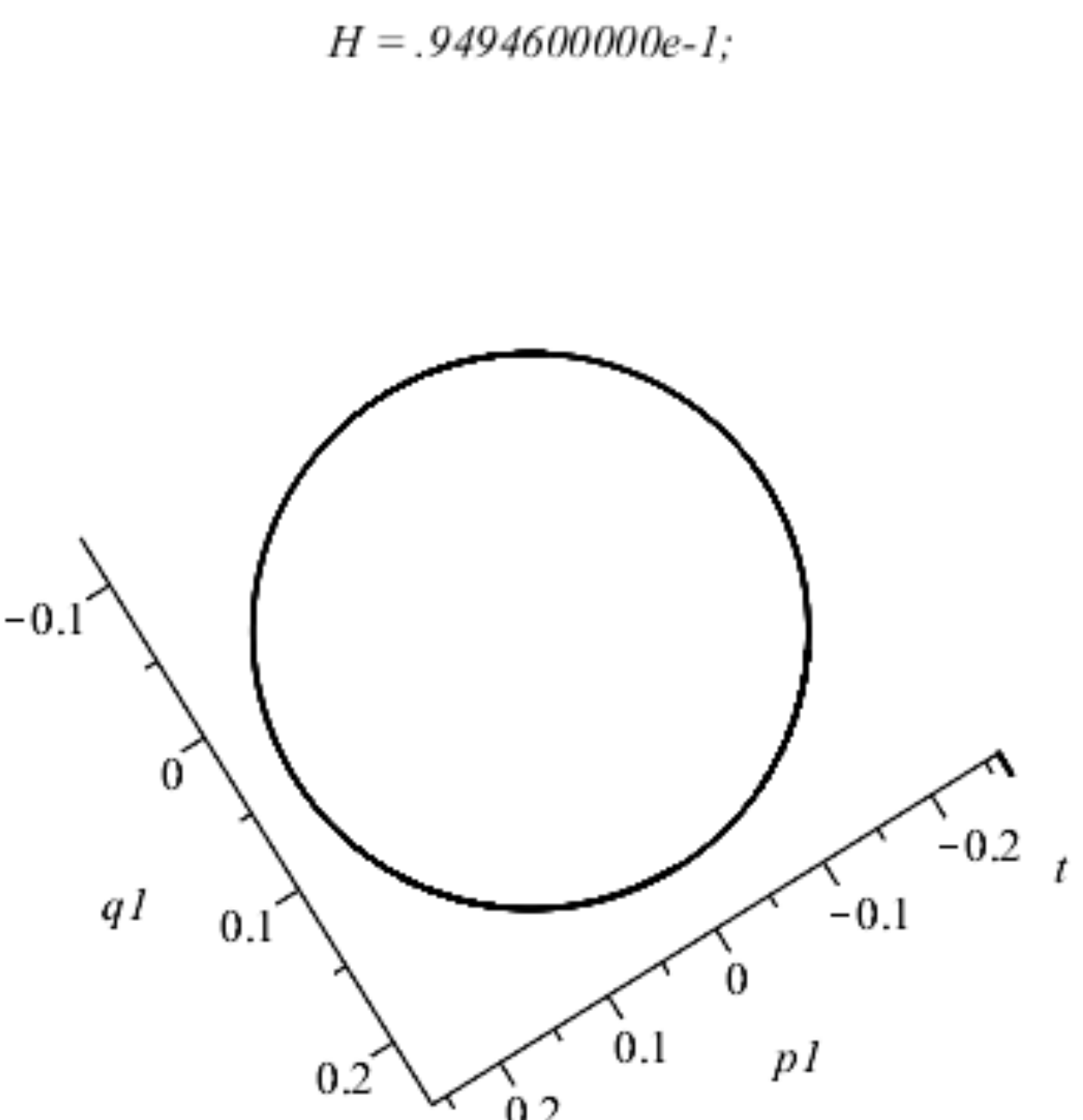}
\includegraphics[width= 4.5cm]{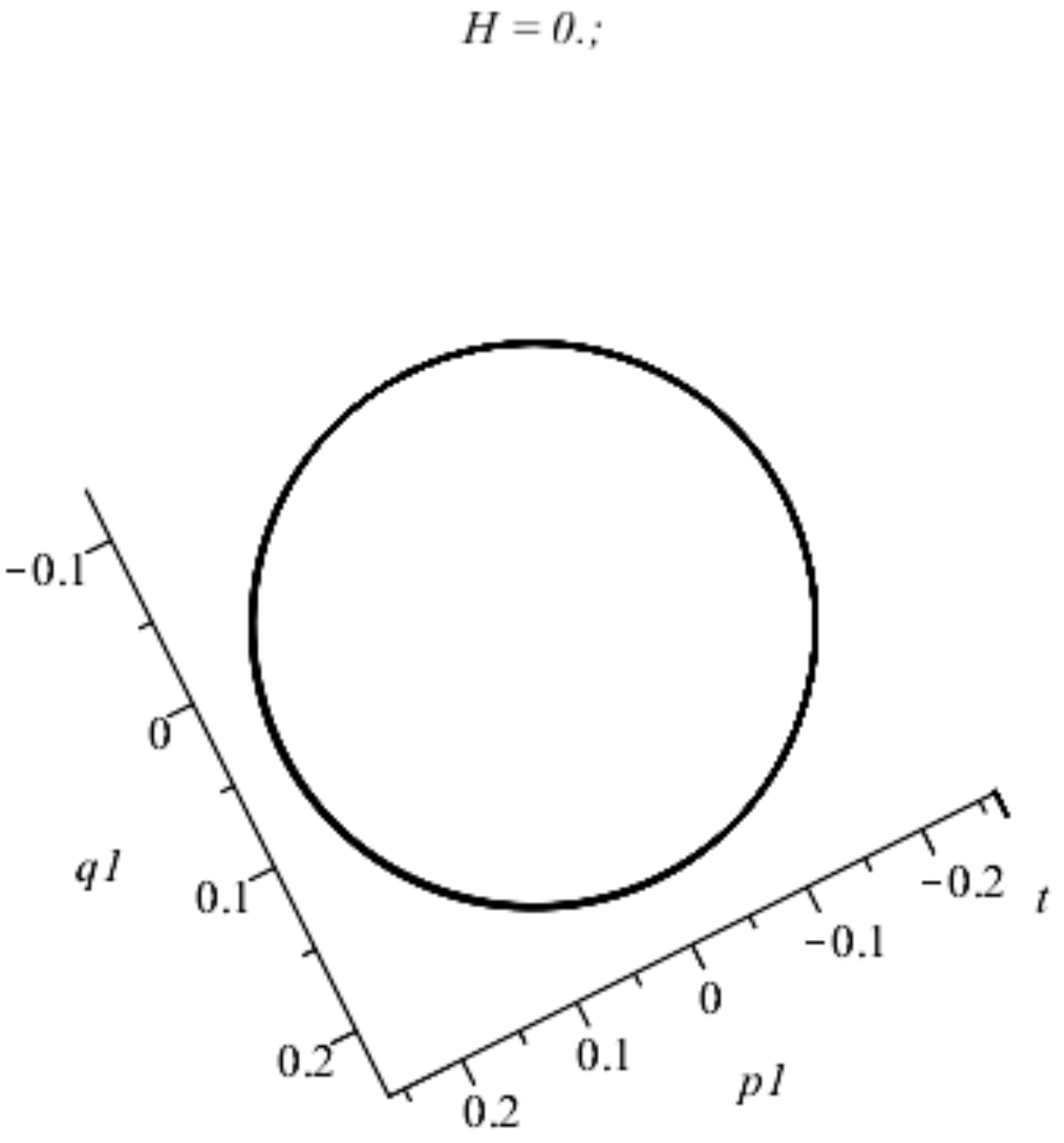}
\caption{Projections of the orbits of $ H$ and $\tilde H_1$ on $\big(p_1,q^1\big)$.}\label{Fig1}
\end{figure}
where we see that the orbits on $\big(q^1,p_1\big)$ of the two systems coincide. But, if we include the dependence on the different Hamiltonian parameters (denoted in both the graphs as $t$), we get Fig.~\ref{Fig2}
\begin{figure}[t]\centering
\includegraphics[width= 4.5cm]{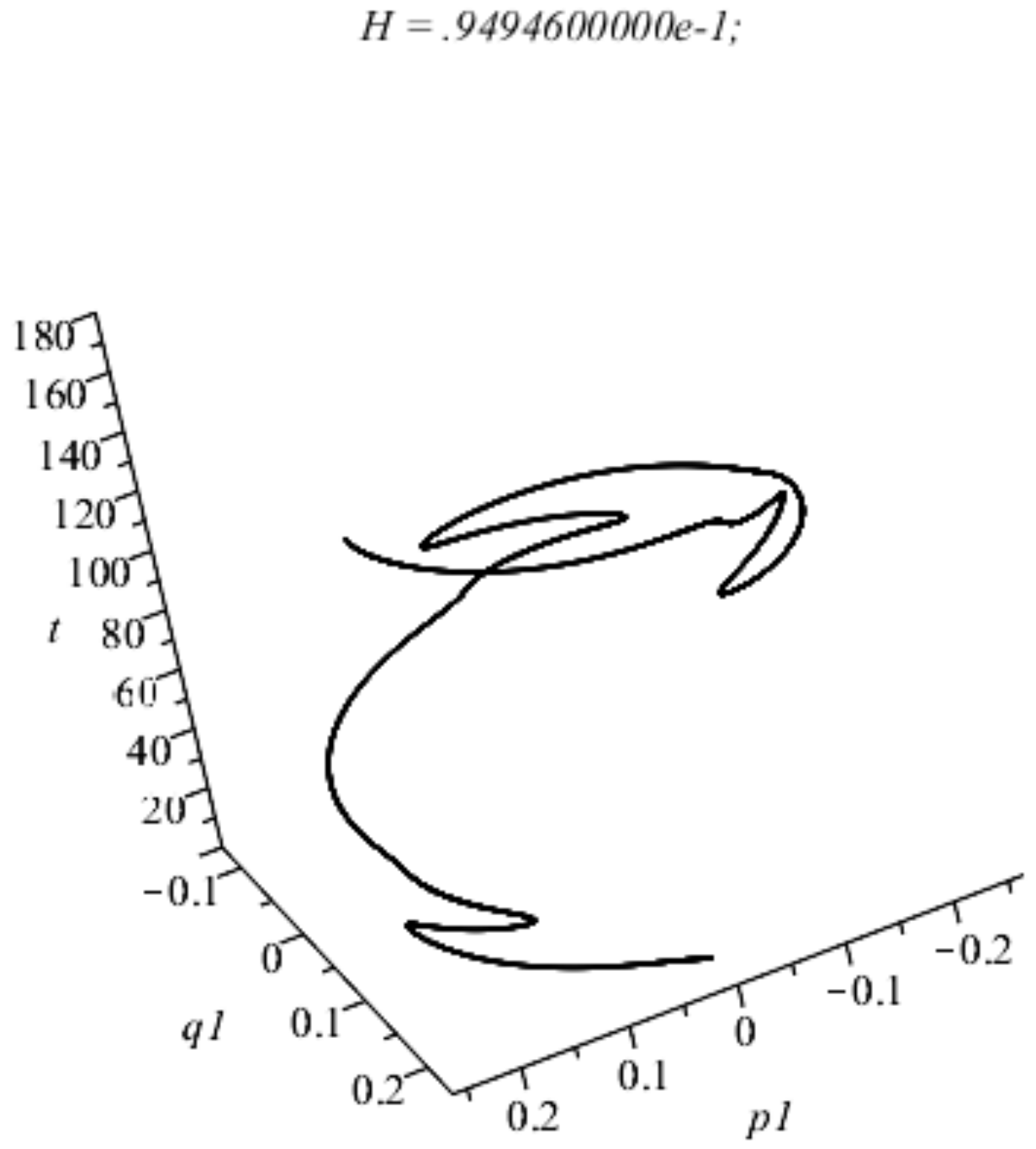}
\includegraphics[width= 4.5cm]{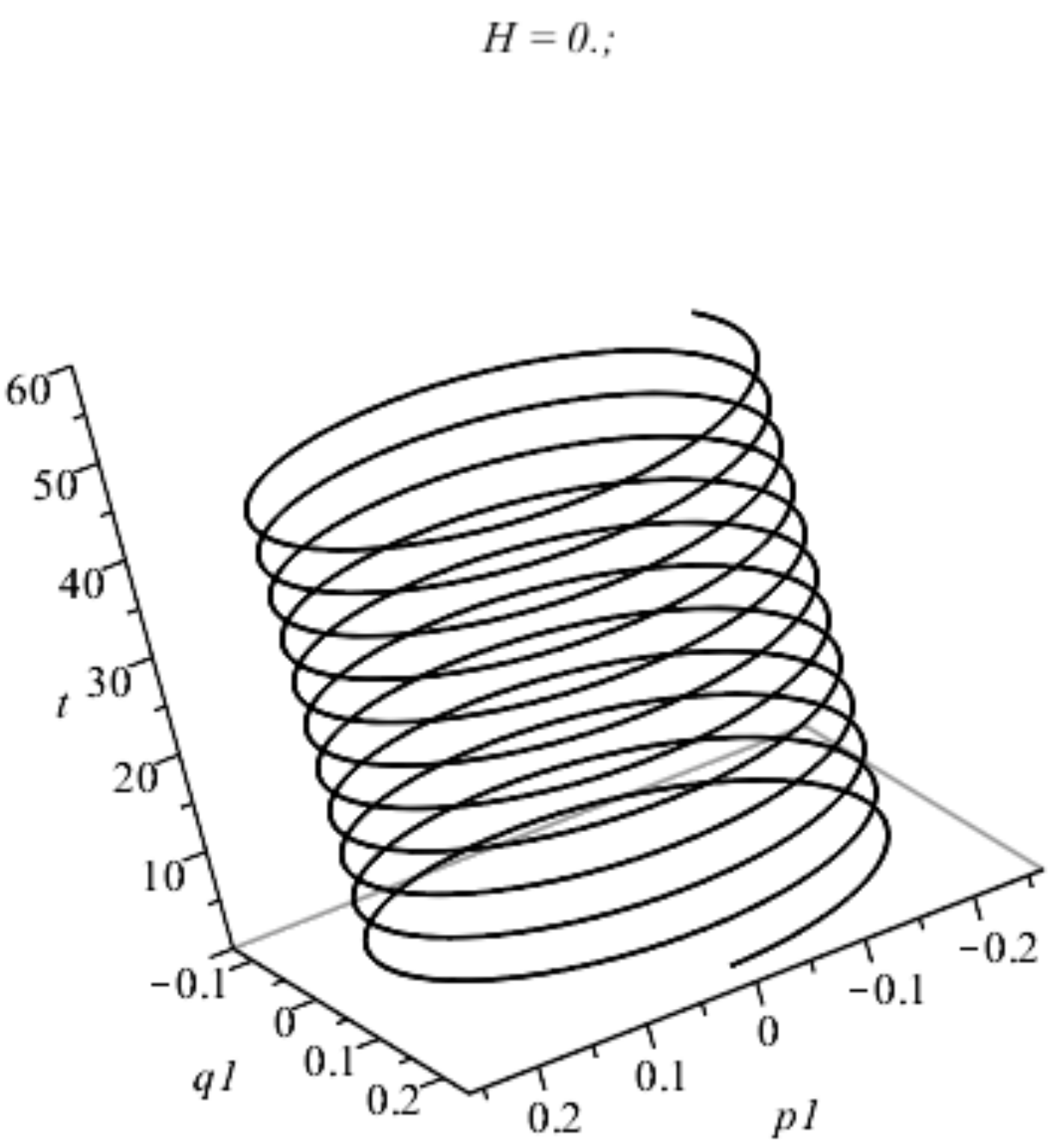}
\caption{Projections of the orbits of $ H$ and $\tilde H_1$ on $\big(t,p_1,q^1\big)$ and $\big(\tau_1,p_1,q^1\big)$.}\label{Fig2}
\end{figure}
and we can see how the dependence on the Hamiltonian parameters can be extremely different in the two cases.
\end{exm}

\begin{exm}\emph{Twisted product with constant coefficients of harmonic oscillators.} Taking twisted products of Hamiltonians seems an interesting way to establish an interaction among Hamiltonian systems. An example, even if somehow trivial, is provided by the twisted product with constant coefficients of harmonic oscillators. Let
 \begin{gather*}H_i=\frac 12 \big(p_i^2+\omega_i^2 \big(q^i\big)^2 \big), \qquad i=1,\dots,n \end{gather*}
be a finite set of harmonic oscillators. Let
 \begin{gather*}H=\alpha^iH_i, \qquad \alpha^i \in \mathbb R^+, \end{gather*}
be their twisted product with constant twist functions. The~$H_i$ are all constants of motion of $H$ and there is actually no interaction among them. However, some effect of the twisted product is nevertheless evident. The Hamilton equations of $H$ are
 \begin{eqnarray*}\frac {{\rm d}q^i}{{\rm d}t}&=&\alpha^i\frac {{\rm d}q^i}{{\rm d}\tau_i}=\alpha^ip_i, \\
 \frac {{\rm d}p_i}{{\rm d}t}&=&-\alpha^i\frac {{\rm d}p_i}{{\rm d}\tau_i}=-\alpha^i\omega_i^2q^i. \end{eqnarray*}
The general solution of these equations is
 \begin{gather*}
q^i(t)=c_1^i\sin\big(\alpha^i \omega^i t\big)+c_2^i\cos \big(\alpha^i \omega^i t\big),\qquad c^i_j\in \mathbb R.
 \end{gather*}
We see that, for example, the choice $\alpha^i=k/\omega_i$, for any real positive $k$, determines a time-scaling that gives to all the oscillators the same frequency $k$ with respect to $t$ (as well as any other real positive frequency for different choices of $k$ for each $i$). Namely, the rescaling is in this case
 \begin{gather*}t=\sum _{i=1}^n\frac 1{\alpha^i}\tau_i+t_0,\qquad \text{or}\qquad \tau_i=\alpha^i t+\tau_i^0. \end{gather*}
The frequency of each oscillator $H_i$ with respect to its own Hamiltonian parameter $\tau_i$ remains clearly $\omega_i$.
\end{exm}

\section{Block-separation} \label{section5}

The results of the previous section can be generalized as follows, leading to a kind of partial separation of variables that we call block-separation.

Let $M$ be a $N$-dimensional manifold. Let us consider a partition of a coordinate system on~$M$ organized as follows. For $n\leq N$, consider for each integer $r=1,\dots,n$ the integers $n_r$ such that
 \begin{gather*}
N=n_1+\dots+n_n.
 \end{gather*}
The coordinate system is therefore composed of $n$ blocks, and for each $r\leq n$ we have an $r$-block of coordinates that we denote as $\big(q^{r_1}, \dots, q^{r_{n_r}}\big)$. We call $M_r$ the manifold spanned by the $r$-block of coordinates. Consider $T^*M$ with the conjugate $r$-block momenta $(p_{r_1}, \dots, p_{r_{n_r}})$.

Let us consider
\begin{gather*}
H_r=\frac 12 g^{r_ir_j}_rp_{r_i}p_{r_j}+V_r\big(q^{r_k}\big),
\end{gather*}
and the $n$ block-separated equations
\begin{gather*}
H_r=S_r^ac_a,\qquad a=1,\dots,n,
\end{gather*}
where we assume that $g_r^{r_ir_j}$, $S_r^a$ and $V_r$ are functions of coordinates of the $r$-block only and $c_a$ are constants. If we assume that the $n\times n$ matrix $(S_r^a)$ is invertible, and we call it {\it block-St\"ackel matrix}, then we can write the $n$ equations
\begin{gather}\label{11}
\big(S^{-1}\big)_a^rH_r=c_a.
\end{gather}
We denote $\alpha^r=\big(S^{-1}\big)_1^r$ and call
\begin{gather*}
H=\alpha^rH_r, \qquad K_a=\big(S^{-1}\big)_a^rH_r, \qquad a=2,\dots, n.
\end{gather*}
Hence, $H$ is in the form of twisted product and it is a natural Hamiltonian whose metric tensor~$G$ is block-diagonalized, with components
 \begin{gather*}
G^{r_ir_j}=\alpha^rg^{r_ir_j}_r,\qquad G^{r_is_j}=0, \qquad s\neq r,
 \end{gather*}
and whose scalar potential has the form
 \begin{gather*}
V=\alpha^rV_r,
 \end{gather*}
while the scalar potentials in $K_a$
are
\begin{gather}\label{V}
W_a=\big(S^{-1}\big)_a^rV_r.
\end{gather}

A necessary condition for the procedure of above is that the (\ref{11}) are indeed constants of motion of $H$.

\begin{prop}\label{p4}The $n$ functions $(H,K_a)$ are all independent, quadratic in the momenta and pairwise in Poisson involution.
\end{prop}
\begin{proof} Since $S$ is a block-St\"ackel matrix, in analogy to (\ref{5bis}) we have
\begin{gather*}
\partial_{r_i}\alpha^s_aS_s^j=\partial_{r_i}\big(\alpha^s_aS_s^j\big)-\alpha^s_a\partial_{r_i}S^j_s=\partial_{r_i}\big(\delta^j_1\big)-\alpha^r_a\partial_{r_i}S^j_r=-\alpha^r_a\partial_{r_i}S^j_r,
\end{gather*}
where $r$ is not summed and $\alpha^r_a=\big(S^{-1}\big)^r_a$.
Then, from the definition of Poisson bracket and of block-separated coordinates, we get the statement.
\end{proof}

For the dynamics of $H$, an analogue of Proposition \ref{p1} holds.

\begin{prop} The dynamics of $H$ coincides in each $r$-block with the dynamics of
 \begin{gather*}\tilde H_r=H_r-c_aS^a_r, \end{gather*}
 up to a reparametrization of the Hamiltonian parameter given by $\alpha^r={\rm d}\tilde \tau_r/{\rm d}t$. So, the projections of the orbits of $H$ on each $T^*M_r$ coincide with the orbits of $\tilde H_r$.
\end{prop}

\begin{proof} The proof follows the same reasoning of the proof of Proposition~\ref{p1}. It follows that, denoting with $(X_H)_r$ the $r$-block component of the Hamiltonian vector field $X_H$ of $H$,
 \begin{gather*}
(X_H)_r=\dot q^{r_i}\partial_{r_i}+\dot p_{r_i}\partial^{r_i},
 \end{gather*}
we have
\begin{gather*}
(X_H)_r=\alpha^rX_{\tilde H_r},
\end{gather*}
where $X_{\tilde H_r}$ is the Hamiltonian vector field of $\tilde H_r$. Therefore, if $(X_H)_r$ is tangent to any submanifold $f\subseteq T^*M$, then also $X_{\tilde H_r}$ is, and vice-versa.
Hence, just as for the St\"ackel systems, the dynamics of $H$ is determined in each $r$-block, up to reparametrizations of the Hamiltonian parameter, by the dynamics of the $\tilde H_r$, with the difference that the $\tilde H_r$ are no longer one-dimensional.
\end{proof}

In this way, the time-independent dynamics of $H$ can be exactly decomposed into the $n$ lower-dimensional separated dynamics of Hamiltonians $\tilde H_r$. The $\tilde H_r$ share with the $H_r$, factors of the twisted product $H$, the same inertial terms, while the scalar potential is modified by the addition of the term $-c_aS^a_r$.

Partial separation of Hamilton--Jacobi equation was introduced by di Pirro in \cite{dP} and gene\-ra\-lized by St\"ackel in~\cite{St97}. He introduced the $n\times n$ matrix $S$ and his results are analogue to our Proposition \ref{p4}. St\"ackel obtained sufficient conditions for partial separation of the Hamilton--Jacobi equation of natural Hamiltonians. His work has been extended more recently in~\cite{Ha} and~\cite{Ma}, including the study of partial separation of the Schr\"odinger equation, obtaining again sufficient conditions for partial separation, and a more detailed form of the components of the metric tensor in partially separable coordinates. We remark that, by introducing twisted products, our characterization of block-separation provides necessary and sufficient conditions for it, in analogy with St\"ackel theory of complete separation. We do not make here a strict comparison between our results and those of \cite{St97} and \cite{Ha}, since these last results are strictly related to Hamilton--Jacobi theory and there is no consideration of the dynamical relations among the $N$-dimensional Hamiltonian and the separated Hamiltonians, which is our main interest. The detailed characterization of the partially separable metric's components in~\cite{Ha} should eventually coincide with a similar characterization of block-separable metrics. We do not consider here the distinction between linear and quadratic in the momenta first integrals (from linear first integrals one can always obtain quadratic ones). It is remarkable that in the last century very few works have been devoted to partial separation of Hamilton--Jacobi equation. This is understandable when one considers that Hamilton--Jacobi theory is of not easy application, apart the simplest cases, even when completely separated integrals of the Hamilton--Jacobi equations do exist. Some applications of partially separated integrals of the Hamilton--Jacobi equation, in order to generate new {\it possible} first integrals of the Hamiltonian, are presented in~\cite{Le} and~\cite{Ni}. Our approach based upon the block-separated dynamics, instead of the partially separated Hamilton--Jacobi equation, appears to be completely new and could be more suitable for applications of the theory, particularly in the analysis of systems with many degrees of freedom.

\subsection{Block-Eisenhart and block-Levi-Civita equations}\label{section5.1}
We can see, with some surprise, that the characterisation of block-separation includes tools developed for St\"ackel complete separation. Indeed, we can formulate classical results by Eisenhart and Levi-Civita in block form. If we assume that $\big(q^i\big)$ are twisted coordinates $\big(q^{r_i}\big)$, then for $G$ we have
 \begin{gather*}
G_{r_ir_j}=\frac 1{\alpha^r}g_{r_ir_j}^r,\qquad r \ \text{not summed},
 \end{gather*}
so that
 \begin{gather*} G_{r_k a}G^{a s_j}=\delta_{r_k}^{s_j}. \end{gather*}

Moreover, if we assume that the coordinates $\big(q^{r_i}\big)$ are block-separated
 \begin{gather*}
k^{r_ir_j}_a=\big(S^{-1}\big)^r_ag^{r_ir_j}_r,
 \end{gather*}
where $k^{bc}_a$ are the components of the 2-tensor associated with $K_a$. Then,
 \begin{gather*}
(k_a)^{r_i}_{r_j}=\frac 1{\alpha^r}\big(S^{-1}\big)^r_a\delta^{r_i}_{r_j},
 \end{gather*}
and we can consider the functions
 \begin{gather*}
\lambda _a^r=\frac 1{\alpha^r}\big(S^{-1}\big)^r_a,\qquad r\; \text{not summed},
 \end{gather*}
as the analogue of eigenvalues of Killing tensors $k_a$ in St\"ackel theory. They are indeed the eigenvalues with respect to $G$ of the Killing tensors $(k_a)$ associated with the first integrals $(K_a)$.

\begin{prop} In block-separated coordinates, we have that
 \begin{gather*}
\{H,K_a\}=0,
 \end{gather*}
if and only if the {\rm block-Eisenhart equations}
\begin{gather}\label{Ee}
\partial_{r_k}\lambda^s_a=\big(\lambda_a^r-\lambda_a^s\big)\partial_{r_k}\ln |\alpha^s|,
\end{gather}
hold, with $r$, $s=1,\dots, n$, for all $r_i$, $s_j$ in the respective separated blocks, and \eqref{V} hold.
\end{prop}

\begin{proof} By expanding
 \begin{gather*}
\{H,K_a\}=0,
 \end{gather*}
in block-separable coordinates, collecting homogeneous terms in the momenta and dividing by~$\alpha^r\alpha^s$, from the higher order terms in the momenta we have,
\begin{gather}\label{ee}
\big(\lambda^s_a-\lambda^r_a\big)\partial_{r_k} g^{s_is_j}+g^{s_is_j}\big[\partial_{r_k}\lambda^s_a-\big(\lambda_a^r-\lambda_a^s\big)\partial_{r_k}\ln |\alpha^s|\big]=0,
\end{gather}
for all $r$, $s$, $r_k$, $s_i$, $s_j$ in the respective blocks. If $r=s$ the equations become
 \begin{gather*}
g^{s_is_j}\partial_{s_k}\lambda^s_a=0,
 \end{gather*}
and, if $r\neq s$, then $\partial_{r_k} g^{s_is_j}=0$. Hence, (\ref{ee}) is equivalent to
 \begin{gather*}
g^{s_is_j}\big[\partial_{r_k}\lambda^s_a-\big(\lambda_a^r-\lambda_a^s\big)\partial_{r_k}\ln |\alpha^s|\big]=0.
 \end{gather*}
If $g^{s_is_j}=0$, the equations are identically satisfied, otherwise, we have~(\ref{Ee}). Since not all the $g^{s_is_j}$ are zero, we have the statement. The first-order terms in the momenta vanish if and only if~(\ref{V}) hold.
\end{proof}

By definition of the $\lambda_a^s$ and of the $\alpha^s$, we have the equations
 \begin{gather*}
S^s_r\alpha^r=\delta^s_1, \qquad S^s_r\lambda_a^r\alpha^r=\delta^s_a.
 \end{gather*}
We observe that, after putting $\alpha^r=g^{rr}$, the previous equations are identical to the relations typical of St\"ackel systems. In the same way, the block-Eisenhart equations become the standard Eisenhart equations.

\begin{prop} The block-Eisenhart equations \eqref{Ee} hold if and only if $(S^r_a)$ is a block-St\"ackel matrix.
\end{prop}

As for the St\"ackel systems, the block-Levi-Civita equations can be considered as the integrability conditions of the block-Eisenhart equations. The derivation is essentially the same as in~\cite{Be}.

It is therefore straightforward to see that the block-diagonalized coordinates $\big(q^1,\dots, q^N\big)$ are block-separated for the Hamiltonian $H$ if and only if the block-Levi-Civita equations
\begin{gather*}\label{bLC}
\alpha^{r}\alpha^{s}\partial_{r_is_j} \alpha^m-\alpha^{r}\partial_{r_i} \alpha^{s}\partial_{s_j} \alpha^m-\alpha^{s}\partial_{s_j} \alpha^{r}\partial_{r_i} \alpha^m=0,
\end{gather*}
are satisfied, where the coordinates $r_i$, $s_j$ are in different blocks and $m=1,\dots,n$. The scalar potential $V$ is already in the form of a block-St\"ackel multiplier, thanks to the form of $H$, and satisfies
\begin{gather*}\label{bLC1}
\alpha^{r}\alpha^s\partial_{r_is_j} V-\alpha^r\partial_{r_i} \alpha^s\partial_{s_j} V-\alpha^s\partial_{s_j} \alpha^r\partial_{r_i} V=0.
\end{gather*}
See (\ref{LC}) for a comparison.

\subsection{Invariant characterization}\label{section5.2}

As in St\"ackel theory, we can use the previous results for an invariant characterization of block-separation. Therefore, we have the analogue of the Eisenhart--Kalnins--Miller--Benenti theorem~\cite{Be}.

\begin{prop}\label{p6} The twisted Hamiltonian $H=\alpha^rH_r$, $r=1,\dots,n$, is block-separated in twisted coordinates $\big(q^{r_i}\big)$ if and only if
\begin{enumerate}\itemsep=0pt
\item[$1)$] there exist other $n-1$ independent quadratic in the momenta functions $K_a=\alpha_a^rH_r$ such that
 \begin{gather*}
\{H,K_a\}=0,
 \end{gather*}
\item[$2)$] the Killing two-tensors $(k_a)$ are simultaneously block-diagonalized and have common normally integrable eigenspaces
\end{enumerate}
Moreover, it follows that $\{K_a,K_b\}=0$.
\end{prop}

\begin{rmk} The proof of the previous statement relies on the assumption that each $H_r$ depends only on coordinates in $T^*M_r$. Otherwise, if the metric tensor of $H$ is only block-diagonal in $\big(q^{r_i}\big)$, conditions 1) and 2) are only necessary.
\end{rmk}

Consequently, the strategy for finding block-separated coordinates of a given $N$-dimensional natural Hamiltonian $H$ is the following
\begin{itemize}\itemsep=0pt
\item Find a number $n\leq N$ of independent quadratic first integrals $(K_a)$ of $H$ in involution among themselves, whose associated Killing tensors $(k_a)$ admit common block-diagonalized normally integrable eigenspaces. The number $n$ corresponds to the number of blocks. The dimension of the common eigenspaces equals the dimension of each block.
\item At this point, we have block-diagonalized coordinates and we can write $H=\alpha^r H_r$ for some functions $H_r$, the $\alpha^r$ being determined by the block-St\"ackel matrix determined by~$(k_a)$.
\item The last step consists in checking that each $H_r$ depends only on coordinates in~$T^*M_r$.
\end{itemize}
This is indeed the procedure applied in Example~\ref{4c}.

We recall (see \cite{Eis}) that in a Riemannian manifold any symmetric 2-tensor is pointwise diagonalizable, that is there exist $n$ vector fields $E_i$, pointwise orthogonal eigenvectors of $k$, such that $k=\sum_i \lambda^i E_i\otimes E_i$. The $\lambda^i$ are the roots of the characteristic equation $\det(k-\lambda g)=0$ and the geometric multiplicity of each eigenvalue $\lambda^i$ coincides with its algebraic multiplicity.

In our case, the algebraic multiplicity of each $ \lambda^r_a$ is $n_r$ at least (for some $a$, we can have $\lambda_a^r=\lambda_a^s$, and the algebraic multiplicity is in this case $n_r+n_s$).

Proposition \ref{p6} provides an invariant characterization of block-separable coordinates in terms of what we can call {\it block-Killing--St\"ackel algebras} generated by the $(k_a)$. See \cite{Be, BCRks} for a definition of Killing--St\"ackel algebras.

If we assume that for $N$-dimensional 2-tensors $T^\kappa_\lambda$ to each eigenvalue of algebraic multiplicity~$n_r$ it corresponds a space of $n_r$ linearly independent covariant eigenvectors $\{X_a\}$, we can consider the $(N-n_r)$-dimensional space of vectors $E_{N-n_r}$ such that $\big\langle X_a,E^b\big\rangle =0$, $\forall\, E^b\in E_{N-n_r}$. We assume that $E_{N-n_r}$ is a regular distribution of constant rank~$N-n_r$.

The necessary and sufficient condition for the integrability of the distributions $E_{N-n_r}$ is given by the Haantjes theorem~\cite{Haa}.

\begin{teo}\label{t7}Let $T^\lambda_\kappa$ be a tensor such that to each root with multiplicity $n_r$ of the
characteristic equation belongs a set of $n_r$ linearly independent covariant eigenvectors.
Then the $E_{N-n_r}$ determined by these vectors are integrable if and only if
 \begin{gather*}
H^\kappa_{\nu \sigma}T^\nu_\mu T^\sigma_\lambda-2 H^\sigma_{\nu[\lambda}T^\nu_{\mu]}T^\kappa_\sigma+H^\nu_{\mu\lambda}T^\kappa_\sigma T^\sigma_\nu=0,
 \end{gather*}
where
 \begin{gather*} H^\kappa_{\mu\lambda}=2T^\nu_{[\mu}\partial_{|\nu|}T^\kappa_{\lambda]}-2T^\kappa_\nu \partial_{[\mu}T^\nu_{\lambda]},
 \end{gather*}
is the so-called Haantjes tensor of $T$.
\end{teo}

Therefore, in analogy with the characterization of the St\"ackel separable coordinate systems, we have

\begin{prop} A natural Hamiltonian admits block-separable coordinates only if its metric tensor admits a symmetric Killing $2$-tensor $T$ satisfying the Haantjes theorem and its scalar potential $V$ satisfies ${\rm d}(T{\rm d}V)=0$. We call $T$ the characteristic tensor of the block-separable coordinates.
\end{prop}

The coordinates are therefore divided into $n$ blocks, where $n$ is the number of the pointwise different eigenvalues of $T$, the dimension of each block equals the multiplicity of the corresponding eigenspace.

An analogue result about characteristic Killing tensors of St\"ackel systems is given in~\cite{AMS} making use of theorems due to Tonolo, Schouten and Nijenhuis (see Section~\ref{section2}). The main difference is due to the fact that, in that case, the eigenvalues of the tensors are simple.

The characterization of block-separation via Killing 2-tensors is extremely powerful in view of applications. For example, in any Riemannian manifold of constant curvature, all Killing tensors, of any order, are linear combinations with constant coefficients of symmetric products of Killing vectors, i.e., isometries~\cite{Ka}. Many common computer-algebra softwares include specific commands for the determination of Killing tensors of Riemannian manifolds.

\begin{exm}\label{4c} \emph{The four-body Calogero system.} The $N$-body Calogero system is the Hamiltonian system of $N$ points of unitary mass on a line, subject to the interaction
 \begin{gather*}
V=\sum_{i=1}^{N-1}\sum _{j=i+1}^N(x_i-x_j)^{-2}.
 \end{gather*}
The Hamiltonian is therefore
 \begin{gather*}
H_{(N)}=\frac 12 \sum _{i=1}^Np_i^2+V,
 \end{gather*}
in Cartesian coordinates $\big(x^i\big)$ and it is known to be maximally superintegrable for any $N$ and multiseparable for $N<4$ \cite{BCRks, SRW}.

In constant curvature manifolds, quadratic first integrals $K_a=\frac 12 k^{ij}_ap_ip_j+W_a$ of natural Hamiltonians can be determined in a systematic way. Indeed, see for example~\cite{BCRks}, $K_a$ is a~first integral of $H$ if and only if the functions~$k^{ij}_a$ are the components of a symmetric Killing 2-tensor~$k_a$ and ${\rm d}W_a=k_a{\rm d}V$. It follows that a necessary condition on~$k_a$ for~$K_a$ to be a first integral of~$H$ is ${\rm d}(k_a{\rm d}V)=0$. As in any constant curvature manifold, the generic Killing 2-tensor of $\mathbb E^4$ is a linear combination (depending on 50 real parameters) of symmetric product of pairs of Killing vectors. By imposing the condition ${\rm d}(k_a{\rm d}V)=0$ to the elements of this space, one finds that $H_{(4)}$ admits, other than the Hamiltonian, only two quadratic independent first integrals in involution, and not the three necessary for standard St\"ackel separation. The two quadratic first integrals of $H_{(4)}$ can be chosen as follows
 \begin{gather*}
K_a=\frac 12 k_a^{ij}p_ip_j+W_a,
 \end{gather*}
where $W_a$ are suitable functions that we will make explicit later on, and
 \begin{gather*}
k_1^{ii}=\sum_{j,k} x^jx^k, \qquad j,k=1,\dots, 4, \qquad j<k,\qquad j,k \neq i,\\
k_1^{rs}=\frac 12\left(\big(x^l\big)^2+\big(x^m\big)^2+x^rx^s- \sum_{j<k} x^jx^k\right),
 \end{gather*}
where $j,k=1,\dots, 4$, $l$, $m$, $r$, $s$ all different,
 \begin{gather*}
k_2^{ii}=\sum _{j\neq i}\big(x^j\big)^2,\qquad
k_2^{ij}=-x^ix^j, \qquad i\neq j.
 \end{gather*}
The eigenvalues of $k_1$ are
 \begin{gather*}
\big\{0,\big(x^1\big)^2+\big(x^2\big)^2+\big(x^3\big)^2+\big(x^4\big)^2\big\}
 \end{gather*}
with multiplicity 1, and
 \begin{gather*}
\sum_{i,j}x^ix^j-\frac 12\big(\big(x^1\big)^2+\big(x^2\big)^2+\big(x^3\big)^2+\big(x^4\big)^2\big), \qquad i,j=1, \dots, 4,\qquad i< j,
 \end{gather*}
with multiplicity 2.

The eigenvalues of $k_2$ are $0$, of multiplicity $1$, and
 \begin{gather*}
\big(x^1\big)^2+\big(x^2\big)^2+\big(x^3\big)^2+\big(x^4\big)^2,
 \end{gather*}
of multiplicity 3. Since the tensors of components $k_1^{ij}$ and $k_2^{ij}$ commute as linear operators and the metric is positive definite, they can be diagonalized simultaneously in some coordinate system.

By using some properties of the eigenvalues of Killing tensors \cite{CR}, one finds that these coordinates $(r,\phi_1,\phi_2, \phi_3)$ are spherical and determined by the consecutive transformations~\cite{CDRrc}
\begin{gather*}
z^1=2^{-1/2}\big(x^1-x^2\big),\\ z^2=6^{-1/2}\big(x^1+x^2-2x^3\big),\\
z^3=12^{-1/2}\big(x^1+x^2+x^3-3x^4\big),\\
z^4=2^{-1}\big(x^1+x^2+x^3+x^4\big),
\end{gather*}
and
\begin{gather*}
z^4=r \cos \phi_1,\\ z^3=r \sin \phi_1 \cos \phi_2,\\ z^2=r \sin \phi_1 \sin \phi_2 \cos \phi_3,\\ z^1=r \sin \phi_1 \sin \phi_2 \sin \phi_3.
\end{gather*}

The scalar potential $V$ of $H_{(4)}$ becomes in these coordinates{\samepage
 \begin{gather*}
V=\frac 1{r^2\sin^2\phi_1} f(\phi_2,\phi_3),
 \end{gather*}
where $f(\phi_2,\phi_3)$ is a rather complicated rational function of trigonometric functions of $\phi_1$, $\phi_2$.}

Therefore, the Hamiltonian becomes
 \begin{gather*}
H_{(4)}=\alpha^1H_1+\alpha^2H_2+\alpha^3H_3,
 \end{gather*}
with
 \begin{gather*}
\alpha^1=1,\qquad \alpha^2=\frac 1{r^2},\qquad \alpha^3=\frac 1{r^2 \sin ^2 \phi_1},\\
H_1=\frac 12 p_r^2+V_1,\qquad H_2=\frac 12 p_{\phi_1}^2+V_2,\qquad H_3=\frac 12 \left(p_{\phi_2}^2+\frac 1{\sin ^2 \phi_2}p_{\phi_3}^2 \right)+V_3,
 \end{gather*}
where
 \begin{gather*}
V=\alpha^iV_i,
 \end{gather*}
with $V_1=0$, $V_2=0$, $V_3=f(\phi_2,\phi_3)$.
Moreover,
 \begin{gather*}
K_1=H_2+\frac{1-2\sin^2 \phi_1}{\sin^2 \phi_1}H_3,\qquad K_2=H_2+\frac 1{\sin^2 \phi_1}H_3.
 \end{gather*}
So that the inverse of the block-St\"ackel matrix is
 \begin{gather*}
S^{-1}=\left(\begin{matrix} 1 && \dfrac 1{r^2} && \dfrac 1{r^2 \sin^2 \phi_1} \vspace{1mm}\\
0 && 1 && \dfrac {1-2\sin ^2 \phi_1}{\sin ^2 \phi_1} \vspace{1mm}\\
0 && 1 && \dfrac 1 {\sin^2 \phi_1} \end{matrix} \right),
 \end{gather*}
and the block-St\"ackel matrix
 \begin{gather*}
S=\left(\begin{matrix} 1 && 0 && -\dfrac 1{r^2} \vspace{1mm}\\
0 && \dfrac 1{2\sin ^2 \phi_1} && \dfrac {2\sin ^2 \phi_1-1}{2\sin ^2 \phi_1} \vspace{1mm}\\
0 && -\dfrac 12 && \dfrac 12 \end{matrix} \right).
 \end{gather*}
 Therefore, the block-separated Hamiltonians are
 \begin{gather*}
\tilde H_1=H_1-c_1+\frac 1{r^2}c_3,\\
\tilde H_2=H_2-\frac 1{2\sin ^2 \phi_1}c_2-\frac {2\sin ^2 \phi_1-1}{2\sin ^2 \phi_1}c_3,\\
\tilde H_3=H_3+\frac 12 c_2-\frac 12 c_3.
 \end{gather*}

The dynamics of the original Hamiltonian $H$ is therefore decomposed into three separated blocks, corresponding to the two dynamics of Hamiltonians $\tilde H_1$, $\tilde H_2$, with one degree of freedom, and the two-degrees of freedom dynamics generated by $\tilde H_3$.
\end{exm}

\begin{exm} \emph{Killing tensor with an eigenvalue of multiplicity $N-1$.} If $H$ admits a single quadratic first integral, this one determines block-separable coordinates if it has exactly one eigenvalue of multiplicity one and another one of multiplicity $N-1$, so that we have a $2\times 2$ St\"ackel matrix. Indeed, from block-Eisenhart equations we have
 \begin{gather*}
X_1\lambda_1=0,\qquad X_{2_i}\lambda_2=0, \qquad i=1,\dots,N-1,
 \end{gather*}
where $X_1$ is the eigenvector corresponding to the eigenvalue $\lambda_1$ of multiplicity one and $X_{2_i}$ the eigenvectors of~$\lambda_2$ of multiplicity $N-1$. Hence, provided~$\lambda_1$ is not a constant, we have that the submanifolds $\lambda_1={\rm const}$ are orthogonal to the eigenvector~$X_1$, which is therefore normally integrable. We can put in this case $X_1=\partial_1$ and the block separation is essentially determined by the equations
 \begin{gather*}
\lambda_1\big(q^2, \dots,q^N\big)={\rm const},
 \end{gather*}
moreover $\lambda_2(q^1)$.

We find in this way another (partial) analogy with St\"ackel separation, since in that case, the existence of a single Killing 2-tensor with distinct eigenvalues in dimension two is enough to determine St\"ackel separable coordinates and the eigenvalues themselves, if not constants, generate the separable coordinates.

In \cite{CR} we show that St\"ackel coordinates can be completely determined by the eigenvalues of the associated Killing two-tensors. Part of those results can be easily extended to block-separable systems. However, we leave the analysis of these questions for future researches.
\end{exm}

\section[Block-separable coordinates of $\mathbb E^3$]{Block-separable coordinates of $\boldsymbol{\mathbb E^3}$}\label{section6}

In dimension three, only two types of block-separable coordinates can exist, plus the trivial case of a single three-dimensional block. So, or each block is one-dimensional, and the coordinates are standard separable orthogonal coordinates, or one block is one-dimensional and the other one is two-dimensional. In this last case, by denoting the separable coordinates as $(u,v,w)$, the geodesic Hamiltonian is
 \begin{gather*}
H=\alpha^1 H_1+\alpha^2 H_2,
 \end{gather*}
with
 \begin{gather*}
H_1=g_1(u)p_u^2,
 \end{gather*}
and, since any 2-dimensional Riemannian manifold is locally conformally flat,
 \begin{gather*}
H_2=g_2(v,w)\big(p_v^2+p_w^2\big),
 \end{gather*}
where the choice of local Cartesian coordinates on the manifolds $u={\rm const}$ is not restrictive. Since $g_1$ can always be set equal to 1 by a rescaling of $u$, we can assume $g_1=1$ and call $g_2$ simply $g$.

The corresponding general block-St\"ackel matrix has the form
\begin{gather*}
S=\left( \begin{matrix} S_1^1(u) && S_1^2(u) \\ S_2^1(v,w) && S_2^2(v,w) \end{matrix} \right).
\end{gather*}
Since
 \begin{gather*}\alpha^1=\big(S^{-1}\big)_1^1,\qquad \alpha^2=\big(S^{-1}\big)_1^2, \end{gather*}
the contravariant metric $G$ associated with $H$ has components
 \begin{gather*}
G^{uu}=\frac{S_2^2(v,w)}\Delta, \qquad G^{vv}=G^{ww}=\frac{-S_1^2(u)g}\Delta,
 \end{gather*}
with $\Delta=S_1^1(u)S_2^2(v,w)-S_1^2(u)S_2^1(v,w)$.

We restrict ourselves to the space $\mathbb E^3$ by imposing that the Riemann tensor of the metric $G$ is identically zero. We have
 \begin{gather*}
R_{2132}=-\frac{3 \big(\partial_uS_1^2 S_1^1-\partial_uS_1^1 S_1^2\big)\big(\partial_w S_2^2 S_2^1-\partial_w S_2^2 S_2^2\big)}{4\Delta S_1^2 g},
 \end{gather*}
and
 \begin{gather*}
R_{3123}=-\frac{3 \big(\partial_uS_1^2 S_1^1-\partial_uS_1^1 S_1^2\big)\big(\partial_v S_2^2 S_2^1-\partial_v S_2^2 S_2^2\big)}{4\Delta S_1^2 g}.
 \end{gather*}
Hence, two cases are possible
\begin{enumerate}\itemsep=0pt
\item[i)] $\big(\partial_uS_1^2 S_1^1-\partial_uS_1^1 S_1^2\big)=0$,

\item[ii)] $\big(\partial_w S_2^2 S_2^1-\partial_w S_2^2 S_2^2\big)=0$ and $(\partial_v S_2^2 S_2^1-\partial_v S_2^2 S_2^2)=0$.
\end{enumerate}

\subsection*{Case i}

We have
 \begin{gather*}S_1^1=a S_1^2, \end{gather*}
where $a$ is a constant. We observe that the components of $G$ become
 \begin{gather*}
G^{uu}=\frac{S_2^2}{S_1^2\big(aS_2^2-S_2^1\big)},\qquad G^{vv}=G^{ww}=-\frac g{aS_2^2-S_2^1},
 \end{gather*}
that is, we can write without restrictions
 \begin{gather*}
G^{uu}=h(u)^{-2} f(v,w)^{-2}, \qquad G^{vv}=G^{ww}=l(v,w)^{-2},
 \end{gather*}
with the obvious definitions of $h$, $f$ and $l$.

A further rescaling of $u$ allows to set $h=1$ and $u$ can always be considered as an ignorable coordinate, therefore, associated with a Killing vector $\partial_u$.

The unknown functions are now reduced to two, and we can consider the remaining components of the Riemann tensor. The resulting equations are
\begin{gather}
l\left(\partial_{vv}l+\partial_{ww}l \right)-(\partial_vl)^2-(\partial_w l)^2=0,\nonumber\\
l\partial_{vv}f-\partial_v l \partial_v f+\partial_w l \partial_wf=0,\nonumber \\
l\partial_{ww}f+\partial_v l \partial_v f-\partial_w l \partial_wf=0,\nonumber\\
l\partial_{vw}f-\partial_v l \partial_w f-\partial_w l \partial_vf=0.\label{16}
\end{gather}

\begin{prop} The coordinate leaves $u={\rm const}$ are planes and $\partial_u$ is proportional to a Killing vector.
\end{prop}

\begin{proof}Equation (\ref{16}) means that the Riemann tensor of the coordinate leaves $u={\rm const}$ is zero.
\end{proof}

The block-St\"ackel matrix is in this case
\begin{eqnarray*}
S=\left( \begin{matrix} a && 1 \\ \left(1-\dfrac a{f^2}\right)gl^2 && -\dfrac{gl^2}{f^2} \end{matrix}\right),
\end{eqnarray*}
and the Hamiltonian $H=\alpha_1^1H_1+\alpha_1^2H_2$ and the first integral $K=\alpha_2^1H_1+\alpha_2^2H_2$ are
 \begin{gather*}
H=\frac 1{f^2}p_u^2+\frac 1{l^2}\big(p_v^2+p_w^2\big), \qquad K=\left(1-\frac a{f^2}\right)p_u^2-\frac a{l^2}\big(p_v^2+p_w^2\big).
 \end{gather*}
Since $K=p_u^2-a H$, we have that the first integral associated with this kind of block-separation is simply $p_u^2$. An expected result, because the variable $u$ is ignorable.

 We observe that the solution of the equation (\ref{16}) is
 \begin{gather}\label{32}
 l=c_2 {\rm e}^{\frac {c_0}2 (w^2-v^2)-c_1v-c_3w},
 \end{gather}
 where $(c_i)$ are real constants. The substitution of (\ref{32}) in the remaining equations yields
 \begin{gather}
 \partial_{vv}f+(c_1-c_0v) \partial_v f-(c_3+c_0w) \partial_wf=0,\nonumber\\
 \partial_{ww}f-(c_1-c_0v) \partial_v f+(c_3+c_0w) \partial_wf=0,\nonumber\\
 \partial_{vw}f+(c_3+c_0w) \partial_v f+(c_1-c_0v) \partial_wf=0.\label{eqq}
 \end{gather}
 The sum of the first two equations implies that $f$ is a solution of the Laplace equation. Then,
 \begin{gather*}
 f=f_1(v+{\rm i}w)+f_2(v-{\rm i}w).
 \end{gather*}
 The substitution of this equation into (\ref{eqq}) gives, after few manipulations, the equivalent system
 \begin{gather*}
 f_1''+(c_1-{\rm i}c_3-c_0(v+{\rm i}w))f_1' = 0,\\
 f_2''+(c_1+{\rm i}c_3-c_0(v-{\rm i}w))f_2' = 0.
 \end{gather*}

\begin{exm} \emph{Rotational and cylindrical coordinates.} Given in the Euclidean plane any coordinate system with a symmetry axis, the coordinates of~$\mathbb E^3$ obtained by rotating the plane around the symmetry axis and by taking as third coordinate the angle of rotation, are of this form. If $f$ is constant and $l^{-2}$ represents a metric in the Euclidean plane $(v,w)$, then the cylindrical coordinate system with the plane $(v,w)$ as base is of this form.
\end{exm}

\subsection*{Case ii}

A similar analysis for Case ii gives the equivalent condition
 \begin{gather*}S_2^1=aS_2^2. \end{gather*}
We have
 \begin{gather*}
G^{uu}=\frac{1}{S_1^1-aS_1^2},\qquad G^{vv}=G^{ww}=-\frac {S_1^2g}{S_2^2\big(S_1^1-aS_1^2\big)},
 \end{gather*}
and we can write without restrictions
 \begin{gather*}
G^{uu}=h(u)^2, \qquad G^{vv}=G^{ww}=l(u)^2 f(v,w)^2,
 \end{gather*}
with the obvious definitions of $h$, $f$ and $l$.

Again, a rescaling of $u$ allows to set $h=1$, in this case, $u$ is not necessarily ignorable, but the metric has the structure of a warped metric. Therefore, $\partial_u$ must be parallel to a conformal Killing vector \cite{CDRg}.

The equations arising from the imposition that the Riemann tensor is equal to zero are now
\begin{gather}
l''l-2(l')^2 = 0,\nonumber\\
l^4\big(f(\partial_{vv}f+\partial_{ww}f)-(\partial_vf)^2-(\partial_wf)^2\big)-(l')^2 = 0.\label{def}
\end{gather}

The block-St\"ackel matrix is in this case
\begin{gather*}
S=\left( \begin{matrix} 1-al^2 && -l^2 \\ a \dfrac g{f^2} && \dfrac g{f^2} \end{matrix}\right),
\end{gather*}
and the Hamiltonian $H=\alpha_1^1H_1+\alpha_1^2H_2$ and first integral $K=\alpha_2^1H_1+\alpha_2^2H_2$ are
 \begin{gather*}
H=p_u^2+l^2f^2\big(p_v^2+p_w^2\big),\qquad K=-ap_u^2+ \big(1-al^2\big)f^2\big(p_v^2+p_w^2\big).
 \end{gather*}
Since $K=f^2\big(p_v^2+p_w^2\big)-a H$, we have that the first integral associated with this kind of block-separation is
 \begin{gather*}f^2\big(p_v^2+p_w^2\big), \end{gather*}
as expected, due to the warped form of the metric.

The solutions of (\ref{def}) are
\begin{eqnarray*}
l=-(c_1u+c_2)^{-1},
\end{eqnarray*}
that, substituted in the second equation, give
\begin{gather}\label{curv}
f(\partial_{uu}f+\partial_{ww}f)-(\partial_vf)^2-(\partial_wf)^2-c_1^2=0.
\end{gather}

We remark that (\ref{curv}) means that the Ricci scalar of the submanifolds $u={\rm const}$ is
 \begin{gather*}
R=2l^2(u)c_1^2,
 \end{gather*}
and its Riemann tensor has non null components
 \begin{gather*}
R_{yzyz}=\frac {c_1^2}{f^2}.
 \end{gather*}
Therefore,

\begin{prop} In Case ii, all the coordinate leaves orthogonal to $\partial_u$ are planes $(c_1=0)$ or spheres. The vector $\partial_u$ is proportional to a conformal Killing vector.
\end{prop}

\begin{exm}{\it Spherical and cylindrical coordinates.} An example of Type~ii of block separation is given by spherical-type coordinates, where $u$ is the radius of the spheres and $(v,w)$ are any coordinates on the sphere. If $c_1=0$, $c_2\neq 0$, then the coordinates are cylindrical with the plane $(v,w)$ as base.
\end{exm}

\begin{rmk}From the previous remarks it follows that, if a coordinate system is block-separable in $n>m$ blocks, then, not necessarily it is block separable in $m$ blocks too. Indeed, if we consider the ellipsoidal coordinates in $\mathbb E^3$, they are St\"ackel-separable, then block-separable in three blocks, but they cannot be block-separable in 2 blocks, since these coordinates do not include planes or spheres. This fact is not surprising, since it is known that St\"ackel separation in ellipsoidal coordinates cannot be achieved by successive separation of the single variables, and block-separation into two blocks in dimension three means exactly that one of the variables can be separated from the others.
\end{rmk}

 The classification of 4D block-separable coordinates, even in Euclidean spaces, appears to be much less simple than in 3D spaces. We must leave the classification of block-separable coordinates in 4D Euclidean spaces for other works. Partial separation of Hamilton--Jacobi and Helmholtz equations on 4D Riemannian manifolds is briefly considered in \cite{BKM}.

\section{Conclusions and future directions}\label{section7}

By introducing the idea of twisted product of natural Hamiltonians and the analysis of the consequent relations among the Hamiltonian flows of the product and of its factors, we provide a new, dynamical, interpretation of classical partial separation of variables of Hamilton--Jacobi equation, as well as of complete separation. We find that our block-separation, when possible, allows the reconstruction of the orbits of the product Hamiltonian from the orbits of the several lower-dimensional block-separated Hamiltonians. We characterize block-separation in an invariant way, by adapting classical results of complete separation theory for the Hamilton--Jacobi equation.
A deeper characterisation of the twisted form of the Hamiltonian is in progress.
Furthermore, we extend Eisenhart's classification of completely separable coordinate systems, the St\"ackel systems, in $\mathbb E^3$ to block-separable coordinate systems, finding essentially coordinate-blocks of rotational, cylindrical, and spherical type. We are confident that the possibility of reducing the analysis of the dynamics of Hamiltonians with many degrees of freedom to the dynamics of its lower-dimensional, block-separated Hamiltonians can find many applications, even in the field of numerical computations. We do not consider here the block-separation of Schr\"odinger's and other related equations of mathematical physics. Studies on partial separation of these equations are somehow more developed than those on partial separation of Hamilton--Jacobi equation, probably because, for these equations, it is less relevant the absence of completeness in partial separation and the consequent impossibility of application of the Jacobi's canonical transformation. We will study block-separation of Schr\"odinger and related equations in next papers.

\subsection*{Acnowledgements}
We are grateful to Raymond McLenaghan, Krishan Rajaratnam and Piergiulio Tempesta for conversations useful to clarify some points of our research and to the anonymous referees which gave a relevant contribution to the improvement of this article.

\pdfbookmark[1]{References}{ref}
\LastPageEnding

\end{document}